\pdfoutput=1
\RequirePackage{ifpdf}
\ifpdf 
\documentclass[pdftex]{sigma}
\else
\documentclass{sigma}
\fi

\numberwithin{equation}{section}

\newtheorem{Theorem}{Theorem}[section]
\newtheorem{Corollary}[Theorem]{Corollary}
\newtheorem{Lemma}[Theorem]{Lemma}

 { \theoremstyle{definition}
\newtheorem{Definition}[Theorem]{Definition}

\newtheorem{Remark}[Theorem]{Remark}
\newtheorem{delta-assume}[Theorem]{Assumption}

}

\def\bdelta{\relax W^*} \def\bphi{\boldsymbol{\phi}}
\def\Wd{\relax W} \def\HN{\mathcal{H}^{(N)}} \def\MN{\C_{N\times N}}
\def\CN{\C^N} \def\Dist{\mathbb{D}} \def\C{\mathbb{C}}
\def\N{\mathbb{N}} \def\Z{\mathbb{Z}} \def\Trans#1{\top_{#1}}
\def\Cent{\mathcal{C}} \def\Weyl{\mathcal{W}}
\def\WeylFlat{\mathcal{W}^{\flat}} \def\A{\mathcal{A}}
\def\Aflat{\mathcal{A}^{\flat}} \def\Supp{S} \def\bisp{\flat}
\def\t{\textbf{t}} \def\L{\mathcal{L}} \def\ad{\textup{ad}}
\def\piQK{\hat\pi}

\begin{document}

\allowdisplaybreaks

\newcommand{\arXivNumber}{1505.02833}

\renewcommand{\PaperNumber}{087}

\FirstPageHeading

\ShortArticleName{Bispectrality of $N$-Component KP Wave Functions: A~Study in Non-Commutativity}

\ArticleName{Bispectrality of $\boldsymbol{N}$-Component KP Wave Functions:\\ A~Study in Non-Commutativity}

\Author{Alex KASMAN}

\AuthorNameForHeading{A.~Kasman}

\Address{Department of Mathematics, College of Charleston, USA}
\Email{\href{mailto:kasmana@cofc.edu}{kasmana@cofc.edu}}
\URLaddress{\url{http://kasmana.people.cofc.edu}}

\ArticleDates{Received May 13, 2015, in f\/inal form October 28, 2015; Published online November 01, 2015}

\Abstract{A wave function of the $N$-component KP Hierarchy with
continuous f\/lows determined by an invertible matrix $H$ is
constructed from the choice of an $MN$-dimensional space of
f\/initely-supported vector distributions.  This wave function is shown to
be an eigenfunction for a ring of matrix dif\/ferential operators in $x$
having eigenvalues that are matrix functions of the spectral parameter
$z$.  If the space of distributions is invariant under left
multiplication by $H$, then a matrix coef\/f\/icient
dif\/ferential-translation operator in $z$ is shown to share this
eigenfunction and have an eigenvalue that is a matrix function of~$x$.
This paper not only generates new examples of bispectral operators, it
also explores the consequences of non-commutativity for techniques and
objects used in previous investigations.}


\Classification{34L05; 16S32; 37K10}

\section{Introduction}

The ``bispectral problem'' seeks to identify linear operators $L$ acting on functions of the variable~$x$ and~$\Lambda$ acting on functions of the variable $z$
such that there exists an eigenfunction $\psi(x,z)$ satisfying the equations
\begin{gather*}
L\psi=p(z)\psi\qquad\mbox{and}\qquad \Lambda\psi=\pi(x)\psi.
\end{gather*}
In other words, the components of the \textit{bispectral triple} $(L,\Lambda,\psi(x,z))$ satisfy two dif\/ferent eigenvalue equations, but with the roles of the spacial and spectral variables exchanged.

The search for bispectral triples
 was originally
formulated and investigated by Duistermaat and Gr\"unbaum~\cite{DG} in a paper which
completely resolved the question in the special case in which the
operators were scalar dif\/ferential operators with one being a
Schr\"odinger operator.  Since then, the bispectrality of many dif\/ferent sorts of
operators have been considered and many
connections to dif\/ferent areas of math and physics have also been
discovered.  (See \cite{BispBook} and the articles referenced therein.)

The present paper will be considering a type of bispectrality in which both the operators and eigenvalues behave dif\/ferently when acting from the left than from the right.  It is necessary to introduce some notation and reorder the terms in the eigenvalue equations in order to properly describe the main results.

Throughout the paper,
$M$ and $N$ should be considered to be  f\/ixed (but arbitrary) natural
numbers. In addition,
$H$ is a f\/ixed (but arbitrary) invertible constant $N\times N$ matrix.
Let $\Delta_{n,\lambda}$ denote the linear functional acting from the
right on functions of $z$ by dif\/ferentiating $n$ times and evaluating at
$z=\lambda$:
\begin{gather*}
 (f(z))\Delta_{n,\lambda}=f^{(n)}(\lambda).
 \end{gather*}
  The set of
linear combinations of these f\/initely-supported distributions with
coef\/f\/icients from~$\CN$ will be denoted~$\Dist$:
\begin{gather*}
\Dist=\left\{\sum_{j=1}^{m}\Delta_{n_j,\lambda_j}C_j\colon
m,(n_j+1)\in\N,\ \lambda_j\in\C,\ C_j\in\CN\right\}. 
\end{gather*}
Finally, let $\bdelta\subset\Dist$ be an $MN$-dimensional space of distributions.

The
goal of this paper is to produce from the selection of $H$ and~$\bdelta$ a \relax{bispectral
triple} $(L,\Lambda,\psi)$ such that
\begin{itemize}\itemsep=0pt

\item
$\psi(x,z)=(I+O(z^{-1}))e^{xzH}$ is an $N\times N$ matrix function of $x$ and $z$ with the specif\/ied asymptotics in $z$,

\item $\psi(x,z)(zH)^M$ is
holomorphic in $z$ and in the kernel of every element of $\bdelta$ (i.e.,
$\psi$ ``satisf\/ies the conditions''),

\item
$L\psi(x,z)=\psi(x,z) p(z)$ for the matrix dif\/ferential operator in $x$ acting from the left and some matrix function $p(z)$,
\item and $\psi(x,z)\Lambda=\pi(x)\psi(x,z)$ for a matrix dif\/ferential-translation operator in $z$ acting from the right and some matrix function $\pi(x)$.
\end{itemize}
In the case $N=1$, this goal is already achieved constructively   for
\textit{any} choice of distribu\-tions~\cite{BispSol}.  One interesting
result of the present paper is that for $N>1$, lack of commutativity
with~$H$ may impose an obstacle to f\/inding  such a bispectral triple in
that no such triple exists for certain choices of~$\bdelta$. The three
subsections below each of\/fer some motivation for interest in the existence of such triples.

\subsection{New bispectral triples}

One source of interest in the present paper  is simply the fact that it
generates  examples of bi\-spectral triples that have not previously been studied.  The construction outlined
below produces many bispectral triples $(L,\Lambda,\psi)$, where~$L$ is a
matrix coef\/f\/icient dif\/ferential operator in~$x$, $\Lambda$~is an
operator in~$z$ which acts by both dif\/ferentiation and
\textit{translation} in~$z$, and $\psi(x,z)$ is a matrix function which
asymptotically approaches $e^{xzH}$ for a chosen invertible matrix~$H$.

This can be seen either as a matrix generalization of the paper~\cite{BispSol} which considered exactly this sort of bispectrality in
the scalar case or as a generalization of the matrix bispectrality in~\cite{BGK,BL,WilsonNotes} to the a more general class of eigenfunctions and
operators.

\subsection{Bispectral duality of integrable particle systems}

Among the applications found for bispectrality is its surprising role in
the duality of integrable particle systems.  Two integrable Hamiltonian
systems are said to be ``dual'' in the sense of Ruijse\-naars if the
action-angle maps linearizing one system is simply the inverse of the
action-angle map of the other~\cite{Ruijsenaars}.  When the Hamiltonians
of the systems are quantized, the Hamiltonian operators themselves share
a common eigenfunction and form a bispectral triple~\cite{FGNR}.

Moreover, the duality of the classical particle systems can also be
manifested through bi\-spectrality in that the dynamics of the two
operators in a bispectral triple under some integrable hierarchy can be
seen to display the particle motion of the two dual systems
respectively.  This classical bispectral duality was observed f\/irst in
the case of the self-duality of the Calogero--Moser system~\cite{cmbis,Wilson2}.  In \cite{BispSol}, it was conjectured that
certain bispectral triples involving scalar operators that translate in~$z$ would similarly be related to the duality of the rational
Ruijsenaars--Schneider and hyperbolic Calogero--Moser particle systems.
This was later conf\/irmed by Haine~\cite{Haine}.

The spin generalization of the Calogero--Moser system similarly exhibits
classical bispectral duality~\cite{BGK,WilsonNotes}.  Achieving this result
essentially involved generalizing the scalar case~\cite{cmbis,Wilson2}
to  bispectrality for matrix coef\/f\/icient dif\/ferential operators. It is
hoped that the construction presented in this paper which generalize
that in \cite{BispSol} will similarly f\/ind application to classical
bispectral duality in some future matrix generalization of the results
in~\cite{Haine}.

\subsection{Non-commutative bispectral Darboux transformations}

In the original context of operators on scalar functions, the eigenvalue equations def\/ining bi\-spectrality were originally written with all operators and eigenvalues acting from the left.  However, research into non-commutative versions of the discrete-continuous version of the bispectral problem in non-commutative contexts found it necessary to have operators in dif\/ferent variables acting from opposite sides in order to ensure that they commute \cite{matrix1,matrix3, GPT1,GPT2,GPT3,GPT4} (cf.~\cite{Duran}).
Continuous-continuous bispectrality in the non-commutative context is also a subject of interest \cite{matrix2,matrix3b,matrix4,matrix5}.  As in the discrete case, it was found that the generalizing the results from the scalar case to the matrix case required letting the operators act from opposite sites \cite{BGK,BL,WilsonNotes}.
Building on this observation, the present paper
seeks to further consider the inf\/luence of non-commutativity on
constructions and results already known for scalar bispectral triples.

In this regard, the wave functions of the $N$-component KP hierarchy are
of interest since they asymptotically look like matrices of the form
$e^{xzH}$, where $H$ is an $N\times N$ matrix \cite{BtK,DJKM}.
Consequently, unlike the scalar case or the case $H=I$ considered in
\cite{BGK,WilsonNotes}, the vacuum eigenfunction itself may not commute with the
coef\/f\/icients of the operators. In fact, for the purpose of more fully investigating the consequences of non-commutativity for bispectral Darboux transformations, this paper will go beyond the standard formalism for the $N$-component KP hierarchy by considering the case in which $H$ is not even diagonalizable and therefore has a centralizer with more interesting structure.  Furthermore, following the suggestion of
Gr\"unbaum \cite{matrix3b}, the present paper will consider the case in
which both of the \textit{eigenvalues} are matrix-valued.

By generalizing the construction from \cite{BispSol} to the
context in which the vacuum eigenfunction, eigenvalues, and operator
coef\/f\/icients all generally fail to commute with each other, this
investigation has identif\/ied some results that are surprisingly dif\/ferent than in the commutative case.   For example,
 it is shown that in this context there exist rational Darboux transformations that do \textit{not} preserve the bispectrality of the eigenfunction (see Section~\ref{nogo}) and that bispectral triples do not always exhibit ad-nilpotency (see Remark~\ref{rem:ad-examp}).
These will be summarized in the last section of the paper.

\section{Additional notation}

\subsection{Distributions and matrices}

Let $M$, $N$, $H$ and $\bdelta$ be as in the Introduction.
The set of constant $N$-component column vectors will be denoted by $\CN$
and $\MN$ is  the set of $N\times N$ constant matrices.
Associated to the selection of $H$ one has
\begin{gather*}
\Cent= \{Q\in\MN\colon  [Q,H]=0 \},
\end{gather*}  the
centralizer of $H$ in $\MN$.

Let
 $\{\delta_1,\ldots,\delta_{MN}\}$ be a basis for $\bdelta\subset\Dist$.
Unlike the selection of
$N$, $M$ and $H$ which were indeed entirely arbitrary, two additional
assumptions regarding the choice of $\bdelta$ will have to be made so
that a~bispectral triple may be produced from it.  However, rather than
making those assumptions here at the start, the additional assumptions
will be introduced only when they become necessary.  This should help to
clarify which results are independent of and which rely on the
assumptions.

Nearly all of the objects and constructions below depend on the choice
of the number $N$, the matrix $H$ and the distributions $\bdelta$  that
have been selected and f\/ixed above, but to avoid complicating the
notation the dependence on these selections will not be written
explicitly.  (For instance, the matrix $\Phi$ in \eqref{eqn:Phi} could
be called $\Phi_{N,M,H,\bdelta}$ because it does depend on these
selections, but it will simply be called~$\Phi$.)

\subsection{Operators and eigenvalues}

The operators in $x$ to be considered in this paper will all be
dif\/ferential operators in the variable~$x$ (also sometimes called~$t_1$)
which are polynomials in $\partial=\frac{\partial}{\partial x}$  having
coef\/f\/icients that are $N\times N$ matrix functions of~$x$.  The
operators in~$z$ will be written in terms of
$\partial_z=\frac{\partial}{\partial z}$ or the translation operator
$\Trans{\alpha}\colon f(z)\mapsto f(z+\alpha)$.  More generally, they will be
polynomials in these having coef\/f\/icients that are $N\times N$ matrix
rational functions of~$z$.

Because operator coef\/f\/icients and eigenvalues will be matrix-valued, the
action of an operator will depend on whether its coef\/f\/icients multiply
from the right or the left.  It also matters whether  the eigenvalue
acts by multiplication on  the right or the left of the eigenfunction.
This paper will adopt the convention that all operators in $x$ that are
independent of $z$ (whether they are dif\/ferential operators or simply
functions acting by multiplication) act from the left and that all
operators in $z$ that are independent of $x$ (including functions,
translation operators, f\/initely-supported distributions and dif\/ferential
operators) act from the right. The action of an operator in $z$ will be
denoted simply by writing the operator to the right of the function it
is acting on.  So, for instance, the function $e^{xzH}$ satisf\/ies the
eigenvalue equations
\begin{gather*}
 \partial e^{xzH}=e^{xzH}
(zH)\qquad\text{and}\qquad e^{xzH}\partial_z=(xH)e^{xzH}.
\end{gather*}

The decision to have operators in $x$ and $z$ acting from dif\/ferent
sides is not merely a matter of notation.  The need for such an
assumption for the dif\/ferential operators in $x$ and $z$ respectively
was already noted in prior work on matrix bispectrality \cite{BGK,BL,matrix3,GPT1,GPT2,GPT3,GPT4,WilsonNotes}.
The present work extends this convention to the eigenvalues and
f\/initely-supported distributions as well, and does so because the
theorems fail to be true otherwise.

\begin{Remark}\label{rem:eigenchange} Note that one needs to be cautious
about applying intuition about eigenfunctions in a commutative setting
without considering how non-commutativity may af\/fect it.   For example,
although a non-zero multiple of an eigenfunction in the commutative
setting always remains an eigenfunction with the same eigenvalue, here
there are two other possibilities.  Suppose $L\psi(x,z)=\psi(x,z)p(z)$,
so that $\psi$ is an eigenfunction for~$L$ with eigenvalue~$p$, and that~$g$ is an invertible constant matrix.  Then $g\psi$ may not be an
eigenfunction for $L$ if $[L,g]\not=0$ and more surprisingly
even though $\psi g$ is an eigenfunction for~$L$, the corresponding eigenvalue
changes to~$g^{-1}pg$.
\end{Remark}

\section[Dual construction for $N$-component KP]{Dual construction for $\boldsymbol{N}$-component KP}

The purpose of this section is to produce a matrix coef\/f\/icient pseudo-dif\/ferential operator satisfying the Lax equations of the multicomponent KP hierarchy  introduced by Date--Jimbo--Kashiwara--Miwa \cite{DJKM} and mostly follows the approach of Segal--Wilson \cite{SW}.  The proof methods utilized here are rather standard in the f\/ield of integrable systems.  However, one of the main points of this paper is that some of the novel features of this situation pose unexpected obstac\-les to the standard methods used to study bispectrality.  So, although it is not surprising that the operators produced in this way satisfy these Lax equations, the proofs are presented with suf\/f\/icient detail to ensure that they work despite the non-diagonalizability of~$H$ and the fact that the distributions here are acting from the right.

\subsection[The $N$-component Sato Grassmannian]{The $\boldsymbol{N}$-component Sato Grassmannian}

Let $\HN$ denote the Hilbert space of square-integrable vector-valued
functions $S^{1}\to(\C^N)^{\top}$ where $S^{1}\subset\C$ is the unit
circle $|z|=1$ and $(\C^N)^{\top}$ is the set of complex valued row\footnote{In previous papers the elements of
$\HN$ have been written as column vectors.  However, because the
construction of bispectral operators below is most easily described in
terms of matrix f\/initely-supported distributions in~$z$ acting from the
\textit{right}, they will be written here as $1\times N$ matrices.} $N$-vectors.  Denote by $e_i$ for $0\leq i\leq N-1$ the $1\times N$
matrix which has
the value $1$ in column $i+1$ and zero in the others.  This extends to a
basis $\{e_i\colon  i\in\Z\}$ of $\HN$ for which $e_i=z^ae_b$ when $i=aN+b$
for $0\leq b\leq N-1$. The Hilbert space has the decomposition
\begin{gather*}
\HN=\HN_+\oplus \HN_-,
\end{gather*} where $\HN_+$ is the Hilbert closure of the
subspace spanned by $e_i$ for $0\leq i$ and $\HN_-$ is the Hilbert
closure of the subspace spanned by $e_i$ for $i<0$.

\begin{Definition}The Grassmannian ${\rm Gr}^{(N)}$ is set of all closed
subspaces $V\subset \HN$ such that the orthogonal projections $V\to
\HN_-$ is a compact operator and such that the orthogonal projection
$V\to \HN_+$ is Fredholm of index zero~\cite{BtK,Sato,SW}.
\end{Definition}

The notion of $N$-component KP hierarchy to be considered in this paper
is compatible with, but somewhat dif\/ferent from that addressed by
previous authors as the following remark explains.

\begin{Remark} For the $N$-component KP hierarchy, the construction of
solutions from a point in the Grassmannian usually involves a collection
of diagonal constant matrices~$H_{\alpha}$ ($1\leq \alpha\leq N$) such that powers of $zH_{\alpha}$ inf\/initesimally generate the continuous f\/lows   and $z$-dependent
matrices~$T_{\beta}$ ($1\leq \beta\leq N-1$) that generate discrete
f\/lows (sometimes called ``Schlesinger transformations'') of the
hierarchy \cite{BtK,DJKM}. In the present paper, however, only the
continuous f\/lows generated inf\/initesimally by powers of $z$ times the (not
necessarily diagonal) matrix $H$ selected earlier will be considered.
\end{Remark}

\subsection[A point of ${\rm Gr}^{(N)}$ associated to the selection of distributions]{A point of $\boldsymbol{{\rm Gr}^{(N)}}$ associated to the selection of distributions}

As usual, one associates a subspace of $\HN$ to the choice of $\bdelta$
by taking its dual in $\HN_+$ and multiplying on the right by the
inverse of a matrix polynomial in $z$ whose degree depends on the
dimension of $\bdelta$ (cf.\ \cite{cmbis,nKdV2KP,SW,Wilson} where the
analogous procedure involved dividing by a~scalar polynomial):
\begin{Definition}\label{def:Wd} Let $\Wd\subset \HN$ be def\/ined by
\begin{gather*}
\Wd=\big\{p(z)(zH)^{-M}\colon p(z)\in \HN_+,\ (p)\delta=0\ \text{for} \ \delta\in\bdelta\big\}.
\end{gather*} \end{Definition}
\begin{Lemma} $\Wd\in {\rm Gr}^{(N)}$.
\end{Lemma}

\begin{proof}
The image of $\Wd$ under the projection map onto $\HN_-$ is contained in the f\/inite-di\-men\-sio\-nal subspace spanned by the basis elements $e_i$ for $-MN\leq i\leq -1$.  This is suf\/f\/icient to conclude that the projection map is compact.  The map $w\mapsto w(zH)^M$ from $W$ to $\HN_+$ has Fredholm index $MN$ because it has no kernel and the image is the common solution set of $MN$-linearly independent conditions.  The map from $\HN_+$ which f\/irst right multiplies by $(zH)^{-M}$ and then projects onto $\HN_+$ has index $-MN$ since its kernel is spanned by the basis vec\-tors~$e_i$ with $0\leq i\leq MN-1$ but the image is all of $\HN_+$.  The composition of these maps is the projection from $\Wd$ to $\HN_+$ and so its index is the sum of the indices which is zero.\end{proof}

\subsection[$N$-component KP wave function]{$\boldsymbol{N}$-component KP wave function}

\begin{Definition} Let $\psi_0=\exp\big(\sum\limits_{i=1}^{\infty} t_iz^iH^i\big)$
where $\t=(t_1,t_2,\ldots)$ are the continuous KP time variables, with
the variables $t_1$ and $x$ considered to be identical. Let
$\phi_i(\t)=(\psi_0)\delta_i$ ($1\leq i\leq MN$) be the $\CN$-valued
functions obtained by applying each element of the basis of  $\bdelta$
to $\psi_0$.  Combine them as blocks into the $N\times MN$ matrix
$\bphi=(\phi_1\ \cdots\ \phi_{MN})$ and def\/ine the matrix $\Phi$ as the
$MN\times MN$ block Wronskian matrix
\begin{gather}
\Phi(\t)=\left(\begin{matrix}\bphi\\ \dfrac{\partial}{\partial x}\bphi\\
\vdots\\ \dfrac{\partial^{M-1}}{\partial x^{M-1}}\bphi\end{matrix}\right)
= \left(\begin{matrix}\phi_1&\phi_2&\cdots&\phi_{MN}\\
\phi_1'&\phi_2'&\cdots&\phi_{MN}'\\ \vdots & \vdots&\ddots&\vdots&\\
\phi_1^{(M-1)}&\phi_2^{(M-1)}&\cdots&\phi_{MN}^{(M-1)}\end{matrix}\right
).\label{eqn:Phi}
\end{gather}
\end{Definition}

\begin{delta-assume}\label{assumpA} Henceforth, assume that $\bdelta$
was chosen so that the matrix $\Phi$ in \eqref{eqn:Phi} is invertible
for some values of $x=t_1$ (i.e., so that $\det\Phi\not\equiv0$).
\end{delta-assume}

The following remarks of\/fer two dif\/ferent interpretations of the fact
that Assumption~\ref{assumpA} is necessary here but not for the
analogous result in the scalar case \cite{BispSol}.

\begin{Remark} When $N=1$, the requirement that $\det(\Phi)\not=0$ is
equivalent to the requirement that  $\{\phi_1,\ldots,\phi_M\}$ is a
linearly independent set of functions.  Then, the independence of the
basis of distributions  would already ensure that
Assumption~\ref{assumpA} is satisf\/ied.  However, when $N>1$ the
\relax{block} Wronskian matrix~$\Phi$ can be singular even  if the
functions~$\phi_i$ are linearly independent as functions of~$x$.  (For
example, consider the case $M=1$, $N=2$, $\phi_1=(1\ 1)^{\top}$,
$\phi_2=(x\ x)^{\top}$). \end{Remark}

\begin{Remark} The determinant of $\Phi$ can be interpreted as the
determinant of the projection map from $\psi_0^{-1}(\t)\Wd$ to $\HN_+$.
In other words, it is the $\tau$-function of $\Wd$.  (The proof of this
claim is essentially the same as the proof of Theorem~7.5 in
\cite{nKdV2KP}.)  The $\tau$-function is non-zero when
$\psi_0^{-1}(\t)\Wd$ is in the ``big cell'' of the Grassmannian.  In the
case $N=1$, the orbit of any point~$W$ in the Grassmannian under the
action of $\psi_0^{-1}$ intersects the big cell~\cite{SW}.  In contrast,
for the $N$-component KP hierarchy it is known that there are points in
the Grassmannian whose orbit under the continuous f\/lows never intersect
the big cell~\cite{BtK}.
\end{Remark}

\begin{Definition} Due to Assumption~\ref{assumpA}, we may def\/ine the
dif\/ferential operator $K$ as \begin{equation}
K=\partial^MI-\left(\phi_1^{(M)}\ \cdots\
\phi_{MN}^{(M)}\right)\Phi^{-1}\left(\begin{matrix} I\\ \partial I\\
\vdots\\ \partial^{M-1}I\end{matrix}\right).\label{eqn:K} \end{equation}
\end{Definition}

\begin{Lemma}\label{Klemma} \quad
\begin{enumerate}\itemsep=0pt
\item[$(a)$] The operator $K$ defined in
\eqref{eqn:K} is the unique monic $N\times N$ differential operator of
order~$M$ such that\footnote{The expression $D\bphi=0$ is a convenient
way to write that applying the $N\times N$ differential operator $D$ to
each function $\phi_i$ ($1\leq i\leq MN$) results in the zero vector of~$\CN$.}  $K\bphi=0$.
\item[$(b)$] If $L$ is any $N\times N$ matrix differential
operator satisfying $L\bphi=0$  then $L=Q\circ K$ for some differential
operator~$Q$.
\end{enumerate}
\end{Lemma}

It is easy to check that $K\bphi=0$.  Alternatively,
Lemma~\ref{Klemma}(a) follows from results of Etingof, Gelfand and
Retakh on quasi-determinants \cite{EGR}.   However,
Lemma~\ref{Klemma}(b) is apparently a new result.  Although the lemma
was originally formulated for this paper\footnote{Lemma~\ref{Klemma}(b)
is used  in the proofs of Theorem~\ref{KPLax},
Lemma~\ref{KcommuteswithH} and Theorem~\ref{claimLp}.}, a self-contained
proof is being published separately as \cite{MDO-Note}.

\begin{Remark} The theory of quasi-determinants is not explicitly being
used here but the ope\-rator~$K$ def\/ined in \eqref{eqn:K} could
alternatively be computed as a quasi-determinant of a Wronskian matrix
with $N\times N$ matrix entries \cite{EGR}.  The method of
quasi-determinants was applied to the bispectral problem for matrix
coef\/f\/icient operators in \cite{BL}.  So, in this sense, the use of this
operator $K$ here is a continuation of the approach adopted there.
\end{Remark}

\begin{Definition} Let $ \psi(\t,z)=K(\psi_0)(zH)^{-M}. $  This
function
will play an important role as the eigenfunction
for the  operators in $x$ and (given one additional assumption) $z$ to
be introduced below. \end{Definition}

\begin{Theorem} The function $\psi$ defined above has the following
properties:
\begin{itemize}\itemsep=0pt
 \item $\psi(\t,z)=(I+O(z^{-1}))\psi_0$ where
$I$ is the $N\times N$ identity matrix,
\item $\psi(\t,z)\in W$ for all $\t$ in the domain of $\psi$.
\end{itemize}
Consequently,
$\psi=\psi_{\Wd}$ is the $N$-component KP wave function of the point
$\Wd\in {\rm Gr}^{(N)}$.
\end{Theorem}

\begin{proof} Because $\partial(\psi_0)=\psi_0zH$, applying the monic
dif\/ferential operator $K=\partial^M+\cdots$ to $\psi_0$ produces a
function of the form $K\psi_0=P(\t,z)\psi_0$ where $P$ is a polynomial
of degree $M$ in $z$ with leading coef\/f\/icient $H^M$.  Then
$K\psi_0H^{-M}z^{-M}=(I+O(z^{-1}))\psi_0$ as claimed. It remains to be
shown that $\psi$ is an element of $\Wd$.  By Def\/inition~\ref{def:Wd} it
is suf\/f\/icient to note that for each $1\leq i\leq MN$, $\psi
z^MH^M=K\psi_0$ satisf\/ies
\begin{gather*}
(K\psi_0)\delta_i=K ((\psi_0)\delta_i )=K\phi_i=0.\tag*{\qed}
\end{gather*}
\renewcommand{\qed}{}
\end{proof}

\subsection[Lax equations of the $N$-component KP hierarchy]{Lax equations of the $\boldsymbol{N}$-component KP hierarchy}

\begin{Definition} Let $K^{-1}$ denote the unique multiplicative inverse of the monic dif\/ferential operator $K$ in the ring of matrix-coef\/f\/icient pseudo-dif\/ferential operators and let $\L= K \circ \partial \circ  K^{-1}$ be the
pseudo-dif\/ferential operator obtained by conjugating
$\partial$ by $K$.
\end{Definition}

\begin{Theorem}\label{KPLax} $\L$  satisfies the Lax equations
\begin{gather*}
\frac{\partial}{\partial t_i}\L=[(\L^i)_+,\L]
\end{gather*}
 for each $i\in\N$,
where $\big(\sum\limits_{i=-\infty}^n \alpha_i(\t)\partial^i\big)_+=\sum\limits_{i=0}^n
\alpha_i(\t)\partial^i$.
\end{Theorem}

\begin{proof} Let $\phi=(\psi_0)\delta$ for some $\delta\in\bdelta$.  A
key observation is that for each $i\in\N$:
\begin{gather*}
\frac{\partial}{\partial t_i}\phi =  \frac{\partial}{\partial t_i}
(\psi_0)\delta = \left(\frac{\partial}{\partial t_i}
\psi_0\right)\delta =  \left((zH)^i \psi_0\right)\delta =
\left(\frac{\partial^i}{\partial x^i} \psi_0\right)\delta =
\partial^i (\psi_0)\delta=\frac{\partial^i}{\partial x^i}( \phi).
\end{gather*} Using the fact that $\phi$ satisf\/ies these ``dispersion
relations'' and the intertwining relationship $\L^i\circ K= K\circ
\partial^i$ we dif\/ferentiate the identity $ K(\phi)=0$ (which follows
from Lemma~\ref{Klemma}) by $t_i$ to get
\begin{gather*} 0  =
K_{t_i}(\phi)+ K(\phi_{t_i}) =   K_{t_i}(\phi)+ K(\partial^i \phi)
 =   K_{t_i}(\phi)+\L^i\circ  K(\phi)\\
 \hphantom{0} =   K_{t_i}(\phi)+(\L^i)_-\circ
 K(\phi)+(\L^i)_+\circ  K(\phi).
 \end{gather*}
 Since both $K$ and $(\L^i)_+$ are ordinary dif\/ferential operators (as ``$+$'' denotes the projection onto the subring of ordinary dif\/ferential operators), $(\L^i)_+\circ K$
is an ordinary dif\/ferential operator with a right factor of $K$.  Therefore, $\phi$
is in its kernel and the last term in the sum above is zero.  We may
therefore conclude that
\begin{gather} K_{t_i}(\phi)+(\L^i)_-\circ
K(\phi)=0.\label{eqstar}
\end{gather}
(Note that the second term in
this sum need not be zero since $(\L^i)_-$ is not an ordinary
dif\/ferential operator.)

 Using $\L^i\circ  K =  K\circ \partial^i$ we can split $\L^i$ into its
 positive and negative parts to get
 \begin{gather*}
  (\L^i)_-\circ  K =  K\circ
 \partial^i - (\L^i)_+\circ  K.
 \end{gather*} Since the object on the right is just
 a dif\/ference of dif\/ferential operators we know that $(\L^i)_-\circ  K$
 is a dif\/ferential operator.

According to~\eqref{eqstar}, $\Gamma(\bphi)=0$ where $\Gamma$ is the
ordinary dif\/ferential operator
\begin{gather*}
\Gamma= K_{t_i}+(\L^i)_-\circ K.
\end{gather*}
Then by Lemma~\ref{Klemma}, there exists a dif\/ferential operator~$Q$ so
that $\Gamma=Q\circ K$. However,  $\Gamma$ has order strictly less than
$M$ since the coef\/f\/icient of the $M^{\rm th}$ order term of~$ K$ is constant
by construction and since multiplying by $(\L^i)_-$ will necessarily
lower the order. This is only possible if $Q=0$ and $\Gamma$ is the zero
operator. Hence, $K_{t_i} =-(\L^i)_-\circ  K$. The Lax equation follows
because
\begin{gather*}
\L_{t_i}  =   K_{t_i}\circ \partial \circ
K^{-1}- K\circ \partial \circ K^{-1}\circ K_{t_i}\circ K^{-1}\\
\hphantom{\L_{t_i} }{} =
-(\L^i)_-\circ K\circ \partial \circ K^{-1}+ K\circ \partial \circ
K^{-1}\circ (\L^i)_-\circ K\circ  K^{-1}\\
\hphantom{\L_{t_i} }{}=
[\L,(\L^i)_-]=[(\L^i)_+,\L]. \tag*{\qed}
\end{gather*}
\renewcommand{\qed}{}
\end{proof}

\section[Operators in $x$ having $\psi$ as eigenfunction]{Operators in $\boldsymbol{x}$ having $\boldsymbol{\psi}$ as eigenfunction}

\begin{Remark} From this point onwards, the goal is to determine whether
the wave func\-tion~$\psi(\t,z)$ is an eigenfunction for an operator in
$x=t_1$ with $z$-dependent eigenvalue and vice versa.  The higher
indexed time variables will only complicate the notation.  So, it will
henceforth be assumed that $t_i=0$ for $i\geq 2$.  Then, $\bphi$ and the
coef\/f\/icients of $K$ can be con\-sidered to be functions of only the
variable $x$, and $\psi_0(x,z)=e^{xzH}$ and $\psi(x,z)$  (sometimes
called the ``stationary wave function'') are functions of~$x$ and~$z$.
Unlike the numbered assumptions, this one is made for notational
simplicity only.  Dependence on the KP time variables can be added to
the objects to be discussed below so that all claims remain valid.
\end{Remark}

\begin{Definition}  A distribution $\delta\in\Dist$ can be composed with
$p(z)\in\MN[z]$ by def\/ining  $p(z)\circ \delta$ to be the distribution
whose value when applied to $f(z)$ is the same as that of $\delta$
applied to the product $f(z)p(z)$ for any $f(z)\in\MN[z]$. Associate to
the choice $\bdelta$ of distributions the ring of polynomials
\textit{with coefficients in $\Cent$} which turn elements of $\bdelta$
into elements of $\bdelta$:
\begin{gather*}
 \A=\big\{p(z)\in\Cent[z]\colon  p(z)\circ
\delta\in\bdelta\ \forall\, \delta\in\bdelta \big\}.
\end{gather*} In particular,
for each $p\in\A$ and each basis element $\delta_i$ there exist numbers
$c_j$ such that $p\circ\delta_i=\sum\limits_{i=1}^{MN} \delta_j c_j$.
\end{Definition}

As in the scalar case, elements of $\A$ are stabilizers of the point in
the Grassmannian: if $p\in \A$ then $\Wd p\subset \Wd$.  The main
signif\/icance of the ring $\A$ is that there is a dif\/ferential operator
$L_p$ of positive order satisfying $L_p\psi=\psi p(z)$ for every
non-constant $p\in \A$. Interestingly, unlike the scalar case, it will
be shown that the question of which \textit{constant} matrices are in~$\A$ is also of interest in that whether $\psi$ is part of a bispectral
triple is related to whether~$H$ is an element of~$\A$. Before those
facts are established, however, the following def\/initions and results
show that~$\A$ contains polynomials of every suf\/f\/iciently high degree.

\begin{Definition} Let $\Supp\subset\C$ denote the support of the
distributions in $\bdelta$.  That is, $\lambda\in\Supp$ if an only if
$\Delta_{n,\lambda}$ appears with non-zero coef\/f\/icient for some $n$ in
at least one $\delta_i$.  For each $\lambda\in\Supp$ let $m_{\lambda}$
denote the \textit{largest} number $n$ such that $\Delta_{n,\lambda}$
appears with non-zero coef\/f\/icient in at least one element of $\bdelta$.
The scalar polynomial
\begin{gather*}
p_0(z)=\prod_{\lambda\in\Supp}(z-\lambda)^{m_{\lambda}+1}
\end{gather*} will be
used in the next lemma, in Def\/inition~\ref{def:L0} and also in
Theorem~\ref{mainresult} below.
\end{Definition}

\begin{Lemma} For any $p\in\Cent[z]$, the product $p_0(z)p(z)$ is in
$\A$.  So, $p_0(z)\Cent[z]\subset\A.$ \end{Lemma} \begin{proof} Let
$p\in\Cent[z]$ and $\delta\in\bdelta$.  Then applying the distribution
$p_0p\circ\delta$ to any polynomial $q$ is equal to $(qp_0p)\delta$.
The distribution $\delta$ will act by dif\/ferentiating and evaluating at
$z=\lambda$ for each $\lambda\in\Supp$.  However, $p_0$ was constructed
so that it has zeroes of suf\/f\/iciently high multiplicity at each
$\lambda$ to ensure that this will be equal to zero.  Hence $p_0p\circ
\delta$ is the zero distribution, which trivially satisf\/ies the
criterion in the def\/inition of $\A$. \end{proof}

\begin{Definition}\label{bispinvdef} For $p\in\Cent[z]$ where
$p=\sum\limits_{i=0}^n C_i z^i$ for $C_i\in\Cent$ def\/ine
\begin{gather*}
\bisp^{-1}(p)=\sum_{i=0}^n C_i H^{-i} \partial^i.
\end{gather*}
\end{Definition}

\begin{Lemma} For any $p\in\Cent[z]$, the constant coefficient
differential operator $\bisp^{-1}(p)$ has $\psi_0=e^{xzH}$ as an
eigenfunction with eigenvalue $p(z)$.
\end{Lemma}

\begin{proof} If $p$
is the polynomial with coef\/f\/icients $C_i$ as in
Def\/inition~\ref{bispinvdef} then
\begin{gather*}
\bisp^{-1}(p)\psi_0 =  \left(\sum_{i=0}^n C_i
H^{-i}\partial^i\right)\psi_0 =\sum_{i=0}^n C_i H^{-i}\partial^i\psi_0
 = \sum_{i=0}^n C_i H^{-i}\psi_0(zH)^i\\
 \hphantom{\bisp^{-1}(p)\psi_0}{}  =\psi_0\left(\sum_{i=0}^n C_i
H^{-i}(zH)^i\right) = \psi_0\left(\sum_{i=0}^n C_iz^i\right)
=\psi_0p(z). \tag*{\qed}
\end{gather*}\renewcommand{\qed}{}
\end{proof}

\begin{Theorem}\label{claimLp} For any $p\in \A$ there is an $N\times N$
ordinary differential operator $L_p$ in $x$ satisfying the intertwining
relationship
\begin{gather*}
 L_p\circ  K= K\circ \bisp^{-1}(p)
 \end{gather*}
  and the eigenvalue
equation  $L_p\psi(x,z)=\psi p(z)$.
\end{Theorem}

\begin{proof} Let $p\in \A$. Observe that
\begin{gather*} K\circ
\bisp^{-1}(p)(\phi_i) =  K\circ \bisp^{-1}(p)((\psi_0)\delta_i) =K\big(\big(
\bisp^{-1}(p)\psi_0\big)\delta_i\big) = K(( \psi_0)p(z)\circ\delta_i)\\
\hphantom{K\circ\bisp^{-1}(p)(\phi_i)}{}
=K\left(\sum_{j=1}^{MN} (\psi_0)\delta_jc_j \right) = \sum_{j=1}^{MN}
K(\phi_j)c_j =\sum_{j=1}^{MN} 0\times c_j =0.
 \end{gather*}
 But, that
means that each function $\phi_i$ is in the kernel of $K\circ
\bisp^{-1}(p)$ and hence by Lemma~\ref{Klemma} there is a dif\/ferential
operator $L_p$ such that $ K\circ \bisp^{-1}(p)=L_p\circ K$. This
establishes the intertwining relationship.

Applying $L_p$ to $\psi= K \psi_0(zH)^{-M}$ one f\/inds
\begin{gather*}
L_p\psi = L_p\big( K\psi_0(zH)^{-M}\big) = (L_p\circ  K)\psi_0(zH)^{-M} =  \big(
K\circ \bisp^{-1}(p)\big)\psi_0(zH)^{-M} \\
\hphantom{L_p\psi}{} =  K \psi_0 p(z)(zH)^{-M} =  K
\psi_0 (zH)^{-M}p(z)=\psi p(z). \tag*{\qed}
\end{gather*}
\renewcommand{\qed}{}
\end{proof}

\section[Operators in $z$ having $\psi$ as eigenfunction]{Operators in $\boldsymbol{z}$ having $\boldsymbol{\psi}$ as eigenfunction}

\begin{Definition} For any $\alpha\in\C$ let the translation operator
$\Trans{\alpha}$ act on functions of $z$ according to the def\/inition
\begin{gather*}
(f(z))\Trans{\alpha}=f(z+\alpha).
\end{gather*}
 Further let $\Trans{\alpha}^H=
\sum_{i,j} \Trans{\alpha \gamma_i}(\alpha \partial_z)^jC_{ij}$ where
the matrices $C_{ij}$ and constants $\gamma_i$ are def\/ined by the formula
\begin{gather}
 \exp\big({xH^{-1}}\big)=\sum_{i,j}
\exp({\gamma_i x})x^j C_{ij}\qquad
\text{(with $\gamma_i\not=\gamma_{i'}$ if $i\not=i'$).}
\label{eqn:gammadef}
\end{gather}
 \end{Definition}

\begin{Lemma}
The
differential-translation operator $\Trans{\alpha}^H$ has $\psi_0=e^{xzH}$ as an
eigenfunction with eigenvalue $e^{\alpha x}$:
\begin{gather*}
\psi_0\Trans{\alpha}^H=e^{\alpha x}\psi_0.
\end{gather*}
\end{Lemma}

\begin{proof}
By def\/inition,
\begin{gather*}
 e^{xzH}\Trans{\alpha}^H=e^{xzH}\sum_{i,j} T_{\alpha \gamma_i}(\alpha
\partial_z^j)C_{ij}=e^{xzH}\sum_{i,j} e^{x\alpha \gamma_i H}(\alpha x
H)^jC_{ij}.
\end{gather*}
 Then the claim follows if it can be shown that the sum in the last expression is equal to $e^{\alpha x}$.

However, \eqref{eqn:gammadef} holds
not only for any scalar~$x$ but also when it is replaced by some
matrix which commutes with the matrices $H^{-1}$ and $C_{ij}$ which
appear in it.  In particular, since~$\alpha x H$ commutes with
$H^{-1}$, its commutator with the expression on the left side of
equation~\eqref{eqn:gammadef} is zero.  From this one can determine that
its commutator with each coef\/f\/icient matrix~$C_{ij}$ on the right is
zero as well.  However, replacing~$x$ with~$\alpha x H$ yields the
formula
\begin{gather*}
\exp(\alpha x)=\sum_{i,j}
\exp({\gamma_i \alpha x H})(\alpha x H)^j C_{ij}
\end{gather*}
as needed.
\end{proof}

The goal of this section is to produce operators in $z$ which are matrix
dif\/ferential-translation operators in~$z$ having rational coef\/f\/icients
that share the eigenfunction~$\psi$ with the dif\/ferential operators
$L_p$ in $x$ constructed in the previous section.  In order for the
construction to work, an additional assumption is required:
\begin{delta-assume}\label{assumpB}
Henceforth, it is assumed that
$H\in\A$.  Equivalently, assume that for each $1\leq i\leq MN$ there
exist numbers $c_j$ such that
\begin{gather*}
 H\delta_i=\sum_{j=1}^{MN}c_j \delta_j.
\end{gather*}
\end{delta-assume}

\subsection{The anti-isomorphism}\label{sec:antiiso}

This section will introduce an anti-isomorphism between rings of
operators in~$x$ and~$z$ respectively such that an operator and its
image have the same action on~$\psi_0$. The use of such an
anti-isomorphism as a method for studying bispectrality was pioneered in
special cases in~\cite{Wilson} and~\cite{KR} and extended to a very general commutative context in~\cite{BHY}.  (Additionally, three months after the present paper was posted and submitted, a new preprint by two of the same authors as~\cite{BHY} appeared which seeks to further generalize those results to the non-commutative context~\cite{GHY}.)

\begin{Definition} The two rings of operators of interest to this
construction are
\begin{gather*}
 \Weyl=\bigoplus_{\alpha\in\C}e^{\alpha
x}\Cent[x,\partial] \qquad \hbox{and} \qquad
\WeylFlat=\bigoplus_{\alpha\in\C}\Trans{\alpha}^H\Cent[z,\partial] .
\end{gather*}
\end{Definition}

Note that operators from both rings involve  polynomial coef\/f\/icient
matrix dif\/ferential ope\-ra\-tors which commute with the constant matrix~$H$, but elements of~$\Weyl$ may also include a f\/inite number of factors
of the form $e^{\alpha x}$ while elements of~$\WeylFlat$ may similarly
include factors of~$\Trans{\alpha}^H$ for a f\/inite number of complex
numbers~$\alpha$.

\begin{Definition}  Let $\bisp\colon \Weyl\to\WeylFlat$ be def\/ined by
\begin{gather} \bisp\left(\sum_{i=0}^l\sum_{j=0}^m\sum_{k=1}^n
C_{ijk}e^{\alpha_kx}x^i\partial^j\right)
=\sum_{i=0}^l\sum_{j=0}^m\sum_{k=1}^n\partial_z^i  \Trans{\alpha_k}^H
z^jC_{ijk}H^{j-i},
\label{eqn:bisp}
\end{gather}
where $C_{ijk}\in\Cent$
are the coef\/f\/icient matrices.
\end{Definition}

\begin{Lemma}\label{lem:bisp} For any $L_0\in\Weyl$, the operators $L_0$
and $\bisp(L_0)$ have the same action on $\psi_0=e^{xzH}$:
\begin{gather*}
L_0\psi_0=\psi_0\bisp(L_0).
\end{gather*}
 Moreover, $\bisp$ is an
anti-isomorphism.
\end{Lemma}
 \begin{proof} Since
\begin{gather*} \sum_{i=0}^m\sum_{j=0}^n \sum_k C_{ij}x^ie^{\alpha_k
x}\partial^j(\psi_0)  = \sum_{i=0}^m\sum_{j=0}^n \sum_k
C_{ij}e^{\alpha_k x}x^i(zH)^j(\psi_0)\\
\qquad{}
 = \sum_{i=0}^m\sum_{j=0}^n\sum_k C_{ij}e^{\alpha_k
x}x^i(\psi_0)(zH)^j
 = \sum_{i=0}^m\sum_{j=0}^n\sum_k
C_{ij}e^{\alpha_k x}(\psi_0)\big(H^{-1}\partial_z\big)^i(zH)^j\\
\qquad{}
 = \sum_{i=0}^m\sum_{j=0}^n\sum_k
C_{ij}(\psi_0)\Trans{\alpha_k}^H\big(H^{-1}\partial_z\big)^i(zH)^j
 = \sum_{i=0}^m\sum_{j=0}^n \sum_k
(\psi_0)\big(H^{-1}\partial_z\big)^i\Trans{\alpha_k}^H(zH)^jC_{ij}\\
\qquad{}
 = (\psi_0)\sum_{i=0}^m\sum_{j=0}^n\sum_k
\big(H^{-1}\partial_z\big)^i\Trans{\alpha_k}^H(zH)^jC_{ij}
 = \psi_0\bisp\left(\sum_{i=0}^m\sum_{j=0}^n \sum_k C_{ij}x^ie^{\alpha_k
x}\partial^j\right)
\end{gather*} it follows that $L_0$ and $\bisp(L_0)$
have the same action on $\psi_0$.

The inverse map $\bisp^{-1}$ is found by simply moving the factor of
$H^{j-i}$ to the other side of equation~\eqref{eqn:bisp}, which conf\/irms
that $\bisp$ is a bijection. If $K_1,K_2\in \Weyl$ we have
\begin{gather*}
 (K_2\circ K_1)\psi_0=K_2(K_1(\psi_0))=K_2(\psi_0\bisp(K_1))=(K_2(\psi_0))\bisp(K_1
)\\
\hphantom{(K_2\circ K_1)\psi_0}{}
 =(\psi_0\bisp(K_2))\bisp(K_1)=\psi_0(\bisp(K_2)\circ
\bisp(K_1)).
\end{gather*}
 However, we also know that $(K_2\circ
K_1)\psi_0=\psi_0\bisp(K_2\circ K_1)$.  Then, $\bisp(K_2\circ
K_1)-\bisp(K_2)\circ\bisp(K_1)$ has the $\psi_0$ in its kernel.  The
only operator in $\WeylFlat$ having $\psi_0$ in its kernel is the zero
operator.  Therefore,
\begin{gather} \bisp(K_2\circ
K_1)=\bisp(K_2)\circ\bisp(K_1)\label{eqn:anti-iso}
\end{gather} and
the map is an anti-isomorphism.
\end{proof}

\begin{Remark}\label{rem:antiq} The reader may be surprised to see
``anti-isomorphism'' being used to describe a~map satisfying
equation~\eqref{eqn:anti-iso}.  Generally, a map having this property is
called an \textit{isomorphism}.  There are two reasons this terminology
is being used here.  First, that is the terminology that was used for
the analogous map in previous papers in the scalar setting
\cite{BHY,BispSol,KR}.  More importantly, the fact that the order of
operators here is preserved is merely a consequence of the fact that the
operators in $z$ act from the right while the operators in $x$ act from
the left.  Note that when $K_2\circ K_1$ acts on a function it is $K_1$
that acts f\/irst while when $\bisp(K_2)\circ \bisp(K_1)$ acts it is
$\bisp(K_2)$ that acts f\/irst and it is in this sense that it is an
\textit{anti}-isomorphism.  Nevertheless, it is interesting to note that
in this context the map does preserve the order of operators in a
product.  It may be that had the operators in dif\/ferent variables been written as acting from opposite sides even in the commutative case from the beginning, the map would have been described as an isomorphism
instead. \end{Remark}

\subsection[Eigenvalue equations for operators in $z$]{Eigenvalue equations for operators in $\boldsymbol{z}$}

\begin{Definition}\label{def:L0} Let $L_0$ be the constant coef\/f\/icient
matrix dif\/ferential operator
\begin{gather*}
 L_0=\bisp^{-1}(p_0(z)I).
\end{gather*}
\end{Definition}

\begin{Lemma}\label{lem:L0kills}  $L_0\bphi=0$ and so $L_0=Q\circ K$ for
some $Q$.
\end{Lemma}

\begin{proof} Because differentiation in $x$ and
multiplication on the left commute with the application of the
distributions $\delta_i$ we have
\begin{gather*}
L_0\phi_i=L_0(\psi_0)\delta_i=(L_0\psi_0)\delta_i.
\end{gather*}
 By
Lemma~\ref{lem:bisp}, $L_0\psi_0=\psi_0p_0(z)$.  The polynomial~$p_0$
was chosen so that it has a zero of high enough multiplicity at each
point in the support of $\delta_i$ to guarantee that
\begin{gather*}
 (\psi_0
p_0(z))\delta_i=0.
\end{gather*}
 It then follows from Lemma~\ref{Klemma} that
$L_0=Q\circ K$.
\end{proof}

In the remainder of this construction, $Q$ will denote the dif\/ferential
operator such that $L_0=Q\circ K$.

\begin{Lemma}\label{KcommuteswithH} Given Assumption~{\rm \ref{assumpB}}, the
operators $L_0$, $K$ and $ Q$ commute with $H$.
\end{Lemma}

\begin{proof} Applying the operator $H^{-1}KH$ to $\phi_i$ we see that
\begin{gather*}
H^{-1}KH\phi_i = H^{-1}KH\big(e^{xzH}\big)\delta_i=H^{-1}K\big(e^{xzH}\big)H\delta_i
=H^{-1}K\left(\sum_{j=1}^{MN} c_j\phi_j\right) \\
\hphantom{H^{-1}KH\phi_i}{} = H^{-1}\left(\sum_{j=1}^{MN}
c_jK(\phi_j)\right) =0.
\end{gather*}
However, by Lemma~\ref{Klemma}, $K$ is
the unique monic operator of order~$M$ having all the basis functions~$\phi_i$ in its kernel so $H^{-1}KH=K$.

We also know that $L_0$ commutes with $H$ because $p_0$ is a scalar
multiple of the identity and~$\bisp^{-1}$ which turns $p_0$ into $L_0$
only  introduces additional powers of $H$. Then
\begin{gather*} HQ\circ
K  = H L_0=L_0H =  Q \circ KH=Q \circ HK = QH \circ K.
\end{gather*}
Multiplying this by $K^{-1}$ on the right yields
$HQ=QH$.
\end{proof}

Note that $K$ and $Q$ are probably not in~$\Weyl$, as they may be
\textit{rational} functions in~$x$ and a~f\/inite number of functions of
the form~$e^{\alpha x}$.  However, it is possible to clear their
denominators either by multiplying by a function on the left or by
composing with a function on the right, which motivates the following
def\/inition:

\begin{Definition}
Let $\Aflat$ be the set of $x$-dependent $N\times N$
matrix functions def\/ined as follows:
\begin{gather*}
 \Aflat=\left\{ \piQK\in
\bigoplus_{\alpha\in\C}\Cent[x]e^{\alpha x}\colon \piQK=\pi_Q\pi_K,\
\pi_K(x)K\in\Weyl, Q\circ \pi_K(x)\in\Weyl\right\}.
\end{gather*}
 In other words,
it is the set of zero order elements of~$\Weyl$ which factor as a~product such the right factor times~$K$ and~$Q$ composed with the left
factor are both elements of~$\Weyl$.
\end{Definition}

\begin{Lemma} $\Aflat$ is non-empty and contains matrix functions that
are non-constant in $x$. \end{Lemma} \begin{proof} Note that $K$ and $Q$
are operators that are rational in $x$ and a f\/inite number of terms of
the form $e^x$.  If we let $\pi_K$ be the least common multiple of the
denominators of the coef\/f\/icients in $K$ then $\pi_K K$ is in $\Weyl$
(since we have simply cleared the denominator by multiplication).
Similarly, if we let $\pi_Q$ be a high enough power of the least common
multiple of the denominators of $Q$ then $Q\circ \pi_Q\in\Weyl$.  Then,
$\piQK=\pi_Q\pi_K$ is by construction an element of $\Aflat$.  Since,
$\pi_Q f \pi_K$ is also an element of $\Aflat$ for any order zero
element of $\Weyl$, $\Aflat$ contains non-constant matrix functions.
\end{proof}

The main result is the construction of an operator $\Lambda$ in $z$ with
eigenvalue $\piQK$ when applied to the  wave function $\psi$ for every
$\piQK\in\Aflat$:

\begin{Theorem}\label{mainresult} Let $\piQK=\pi_Q\pi_K\in\Aflat$ where
$\bar K=\pi_K K\in\Weyl$ and $\bar Q=Q\circ \pi_Q\in\Weyl$.  Define
$\Lambda:=(zH)^M\circ (p_0(z))^{-1} \circ \bisp(\bar Q)\circ \bisp(\bar
K)\circ (zH)^{-M}$, then
\begin{gather*}
 \psi\Lambda=\piQK(x)\psi.
\end{gather*}
 \end{Theorem}

\begin{proof}
Write $L_0$ as $L_0=\bar Q\circ
(\pi_K(x)\pi_Q(x))^{-1}\circ \bar K$.
Applying this operator to $\psi_0$ and multiplying each side by $p_0^{-1}(z)$ gives
\begin{gather*}
(\bar Q\circ (\pi_K(x)\pi_Q(x))^{-1}\circ \bar K)\psi_0p_0^{-1}(z)=\psi_0.
\end{gather*}
Moving $\bar K$ to the other side using the anti-isomorphism and applying $\flat(\bar Q)$ to both sides this becomes
\begin{gather*}
\big(\bar Q\circ (\pi_K(x)\pi_Q(x))^{-1}\big)\psi_0(\flat (\bar K)\circ p_0^{-1}(z)\circ \flat(\bar Q))=\psi_0\flat (\bar Q).
\end{gather*}
Moving the last expression to the other side of the equality and moving $\flat(\bar Q)$ in it to the other side, we f\/inally get
\begin{gather*}
\bar Q((\pi_K\pi_Q)^{-1}\psi_0\big(\flat (\bar K)\circ p_0^{-1}\circ \flat (\bar Q)\big)-\psi_0)=0.
\end{gather*}
Note that $\bar Q$ is a dif\/ferential operator in $x$ with a non-singular leading coef\/f\/icient (since it is a~factor of the monic dif\/ferential operator $L_0$).  Hence, its kernel is f\/inite-dimensional.  The only way the expression to which it is applied, a polynomial in~$z$ multiplied by~$e^{xzH}$, could be in its kernel for all~$z$ is if it is equal to zero.  From this, we conclude that
\begin{gather*}
\flat(\bar K)\circ p_0^{-1}\circ\flat( \bar Q) = \flat(\pi_K\pi_Q).
\end{gather*}

Using this one f\/inds that the action of $\Lambda$ on $\psi$ is
\begin{gather*}
(\psi)\Lambda = \big(K\psi_0(zH)^{-M}\big)\Lambda=\big(\pi_K^{-1}(x)\bar K
\psi_0(zH)^{-M}\big)\Lambda\\
\hphantom{(\psi)\Lambda}{}
= \big(\pi_K^{-1}(x)(\psi_0)\bisp{\bar
K}(zH)^{-M}\big)\Lambda
   = \big(\pi_K^{-1}(x)\psi_0\big)\bisp{\bar K}\circ
(zH)^{-M}\circ\Lambda \\
\hphantom{(\psi)\Lambda}{}
= \big(\pi_K^{-1}(x)\psi_0\big)\bisp{\bar K}\circ
(p_0(z))^{-1} \circ \bisp(\bar Q)\circ \bisp(\bar K)\circ (zH)^{-M}\\
\hphantom{(\psi)\Lambda}{}
 = \big(\pi_K^{-1}(x)\psi_0\big)\Lambda_0\circ \bisp(\bar K)\circ (zH)^{-M}
 = \big(\pi_K^{-1}(x)\pi_K(x)\pi_Q(x)\psi_0\big)\circ \bisp(\bar K)\circ
(zH)^{-M}
 \\
\hphantom{(\psi)\Lambda}{}
 = (\pi_Q(x)\psi_0)\circ \bisp(\bar K)\circ (zH)^{-M}
 = \big(\pi_Q(x)\pi_K(x)\pi_K^{-1}(x)\psi_0\big)\circ \bisp(\bar K)\circ
(zH)^{-M}
 \\
\hphantom{(\psi)\Lambda}{}
= \piQK(x)(\pi_K^{-1}(x)\psi_0)\bisp(\bar K)\circ
(zH)^{-M}
= \piQK(x)\big(\pi_K^{-1}(x)\bar K\psi_0\big)
(zH)^{-M}
\\
\hphantom{(\psi)\Lambda}{}
=\piQK(x)(K\psi_0)(zH)^{-M}=\piQK(x)\psi. \tag*{\qed}
\end{gather*}
\renewcommand{\qed}{}
\end{proof}

\subsection{Ad-nilpotency}\label{sec:ad}

We now have eigenvalue
equations $ L_p\psi=\psi p(z)$ and $\psi\Lambda=\piQK(x)\psi $. In the
commutative case, Duistermaat and Gr\"unbaum noticed that these
operators and eigenvalues satisf\/ied an interesting relationship that
goes by the name of ``ad-nilpotency'' \cite{DG}.  As usual, for
operators $R$ and $P$ we recursively def\/ine $\ad_R^nP$ by
\begin{gather*}
\ad_R^1P=R\circ P-P\circ R\ \hbox{and}\ \ad_R^nP=R\circ
\big(\ad_R^{n-1}P\big)-\big(\ad_R^{n-1}P\big)\circ R,\qquad n\geq 2.
\end{gather*}
 When $R$ is a dif\/ferential operator of order greater than one,
then one generally expects the order of $\ad_R^nP$ to get large as~$n$
goes to inf\/inity, but they found that when $R$ is
a scalar bispectral dif\/ferential operator and $P$ is the eigenvalue of a
corresponding dif\/ferential operator in the spectral variable, then
surprisingly $\ad_R^nP$ is the zero operator for large enough $n$.
In fact, in that scalar Schr\"odinger operator case they considered, Duistermaat and Gr\"unbaum found that ad-nilpotency was both a necessary and suf\/f\/icient condition for bispectrality \cite{DG}.

As it turns out, in the this new context even the easier of these two statements fails to hold.   In particular, the ad-nilpotency
in the commutative case is a consequence of the fact that the leading
coef\/f\/icients of $R\circ P$ and $P\circ R$ are equal, but this is not
generally true for dif\/ferential operators with matrix coef\/f\/icients.
Consequently, it is \textit{not} the case that ad-nilpotency holds for
all of the bispectral operators produced by the procedure described
above.  The following results (and an example in
Remark~\ref{rem:ad-examp}) show how and to what extent the old result
generalizes.

\begin{Theorem}\label{ad} Given $L_p$, $\Lambda$, $\piQK(x)$ and $p(z)$
as above, the results of applying $\ad_{\piQK}^nL_p$ and
$\ad_{\Lambda}^n p$ on $\psi$ are equal for every $n\in\N$.
\end{Theorem}

\begin{proof} The case $n=1$ is obtained by subtracting
the equations
\begin{gather*}
 \piQK(x)L_p\psi=\piQK(x)\psi p(z)=\psi\Lambda p(z)
 \end{gather*}
from the equations with the orders of the operators reversed
\begin{gather*}
L_p\piQK(x)\psi=L_p\psi \Lambda=\psi p(z)\Lambda.
\end{gather*} If one assumes the
claim is true for $n=k$ then the case $n=k+1$ is proved similarly and
the general case follows by induction.
\end{proof}

So, the equivalence of the actions of the two operators generated by
iterating ``ad'' remains true regardless of non-commutativity. One
cannot conclude from this alone that either of them is zero without
making further assumptions, but if $\ad_{\piQK}^nL=0$ (which would be
the case, for instance, if one could be certain that $\piQK$ would
commute with the coef\/f\/icients) then because $\psi$ is not in the kernel
of any non-zero operator in $z$ alone, the one could conclude the same
is true for the corresponding operator written in terms of $\Lambda$ and
$p$:

\begin{Corollary}\label{cor}
 If $\ad_{\piQK}^nL=0$ then $\ad_{\Lambda}^n
p=0$. \end{Corollary}

\begin{Remark} This may be the f\/irst time that ad-nilpotency has been
considered for translation operators as well as for matrix operators.
It is therefore relevant to note that Corollary~\ref{cor} is valid even
when the operator $\Lambda$ involves shift operators of the form
$\Trans{\alpha}$  (see Remark~\ref{rem:ad-examp}).
\end{Remark}

\section{Examples}

\subsection{A wave function that is not part of a bispectral triple}\label{nogo}

Consider the case $\bdelta=\operatorname{span}\left\{\delta_1,\delta_2\right\}$,
\begin{gather*}
 H=\left(\begin{matrix}1&0\\0&-1\end{matrix}\right), \qquad
\delta_1=\left(\begin{matrix}\Delta_{1,0}\\\Delta_{0,0}\end{matrix}\right), \qquad
\text{and} \qquad
\delta_2=\left(\begin{matrix}0\\\Delta_{1,0}\end{matrix}\right).
\end{gather*}
 Then Assumption~\ref{assumpA} is met and
 \begin{gather*}
\psi(x,z)=\left(\begin{matrix}\dfrac{xz-1}{xz}e^{xz}&0\vspace{1mm}\\
\dfrac{1}{x^2z}e^{xz}&\dfrac{xz+1}{xz}e^{-xz}\end{matrix}\right).
\end{gather*}

The matrix polynomial $ p(z)=z^2I$ satisf\/ies $p\circ \delta_i=0$ for
$i=1$ and $i=2$ and so $p\in\A$ and as predicted by
Theorem~\ref{claimLp}, the corresponding operator
\begin{gather*}
L_p=\left(\begin{matrix}\partial^2-\dfrac{2}{x^2}&0\vspace{1mm}\\
\dfrac{4}{x^3}&\partial^2-\dfrac{2}{x^2}\end{matrix}\right)
\end{gather*} satisf\/ies
$L_p\psi=\psi p$.

Note that this is a rational Darboux transformation of the f\/irst operator and eigenfunction from the bispectral triple $(\partial^2,\partial_z^2,e^{xzH})$.  Since in the scalar case it has been found that rational Darboux transformations preserve bispectrality, one might expect~$L_p$ and~$\psi$ to be part of a~bi\-spectral triple.
However, $H\delta_i\not\in \bdelta$ and so Assumption~\ref{assumpB} is
not satisf\/ied.  Thus, Theorem~\ref{mainresult} does not guarantee the
existence of a dif\/ferential-translation operator $\Lambda$ in $z$ having
$\psi$ as an eigenfunction. In fact, in this simple case we can see that no such operator
exists.

Brief\/ly, the argument is as follows.  Consider a
dif\/ferential-translation operator $\Lambda$ acting on $\psi$ from the
right and suppose there is a function $\pi(x)$ such that
$\psi\Lambda=\pi(x)\psi$.  By noting the coef\/f\/icients of $e^{xz}$ and
$e^{-xz}$ in the top right entry of each side of this equality, we see
immediately that $\Lambda_{12}=\pi_{12}=0$ (i.e., $\Lambda$ and $\pi$
are both lower triangular.).  Then,we have the scalar eigenvalue
equation $\psi_{11}\Lambda_{11}=\pi_{11}\psi_{11}$ and also
$\psi_{21}\Lambda_{11}=\pi_{21}\psi_{11}+\pi_{22}\psi_{21}$.  Combining
these with  $\psi_{11}+x\psi_{21}=e^{xz}$ one concludes that
$e^{xz}\Lambda_{11}$ can be written in the form $(f(x)+g(x)/z)e^{xz}$.
It follows that $\Lambda_{11}$ as an operator in $z$ has only
coef\/f\/icients that are constant or are a constant multiplied by $1/z$.
In fact, all of the operators in $z$ having $\psi_{11}$ as eigenfunction
are known; they are the operators that intertwine by
$\partial-\frac{1}{z}$ with a constant coef\/f\/icient operator.  The only
ones that meet both the criteria of the previous two sentences are the
order zero operators.  In other words, $\Lambda_{11}$ would have to be a
number.  Then the equation for the action of $\Lambda_{11}$ on
$\psi_{21}$ tells us that $\pi_{21}=0$ and $\pi_{22}=\pi_{11}$ is also a
number.  Since the eigenvalue of this operator $\Lambda$ does not depend
on $x$, the operator does not form a bispectral triple with $L$ and
$\psi$.

\subsection[A rational example with non-diagonalizable $H$]{A rational example with non-diagonalizable $\boldsymbol{H}$}

Consider $H=I+U$,  where
\begin{gather*}
U=\left(\begin{matrix}0&1\\0&0\end{matrix}\right), \qquad \text{so
that}\qquad \psi_0=e^{xzH}=\left(\begin{matrix}e^{xz}&ze^{xz}\\
0&e^{xz}\end{matrix}\right).
\end{gather*} The distributions chosen are
$\bdelta=\operatorname{span}\{\delta_1,\dots,\delta_4\}$ with
\begin{gather*}
\delta_1=\left(\begin{matrix}\Delta_{2,1}\\ 0\end{matrix}\right), \qquad
\delta_2=\left(\begin{matrix}0\\ \Delta_{2,1}\end{matrix}\right), \qquad
\delta_3=\left(\begin{matrix}\Delta_{0,1}\\ 0\end{matrix}\right), \qquad
\text{and}\qquad
\delta_4=\left(\begin{matrix}0\\\Delta_{0,1}\end{matrix}\right).
\end{gather*}
Assumption~\ref{assumpA} is met and the unique monic operator of order
$2$ satisfying $K(\phi_i)=0$ where $\phi_i=(\psi_0)\delta_i$ for $1\leq
i\leq 4$ is
\begin{gather*}
 K=\partial^2 I +\left(\begin{matrix}\dfrac{-2 x-1}{x} & -2
\\ 0 & \dfrac{-2 x-1}{x} \end{matrix}\right)\partial
+\left(\begin{matrix} \dfrac{x+1}{x} & \dfrac{2 x^2+x}{x^2} \vspace{1mm}\\ 0 &
\dfrac{x+1}{x} \end{matrix}\right).
\end{gather*}
 Now  def\/ine
 \begin{gather*}
\psi=K(\psi_0)(zH)^{-2}=\left(\begin{matrix} \dfrac{x z^2-2 x z-z+x+1}{x
z^2} & \dfrac{z-1}{x z^2} \vspace{1mm}\\ 0 & \dfrac{x z^2-2 x z-z+x+1}{x z^2}
 \end{matrix}\right)e^{xzH}.
\end{gather*}

Assumption~\ref{assumpB} is met because $H\delta_i=\delta_i$,
$H\delta_{i+1}=\delta_i+\delta_{i+1}$ for $i=1,3$.  Consequently, it
will be possible to produce operators in $z$ sharing the eigenfunction
$\psi$. The matrix polynomial
\begin{gather*}
p(z)=\left(\begin{matrix}(z-1)^2&(z-1)^3\\
0&(z-1)^2\end{matrix}\right)=(z-1)^2I+(z-1)^3U
\end{gather*} is in $\A$ because
$p\circ\delta_i=2\delta_{i+2}$ and $p\circ\delta_{i+2}=0$ for $i=1,2$.
So,
\begin{gather*}
\bisp^{-1}(p(z))=\left(\begin{matrix}\partial^2-2\partial+1&\partial^3-5
\partial^2+5\partial-1\\ 0&\partial^3-2\partial+1\end{matrix}\right)
\end{gather*}
satisf\/ies the intertwining relationship
\begin{gather*}
 K\circ
\bisp^{-1}(p(z))=L_p\circ K
\end{gather*}
 with
 \begin{gather*}
L_p=\left(\begin{matrix}0&1\cr0&0\end{matrix}\right)\partial^3+
\left(\begin{matrix}1 & -5 \\ 0 & 1  \end{matrix}\right)\partial^2+
\left(\begin{matrix}  -2 & 5-\dfrac{3}{x^2} \vspace{1mm}\\ 0 & -2
\end{matrix}\right)\partial+ \left(\begin{matrix}1-\dfrac{2}{x^2} &
-1+\dfrac{7}{x^2}+\dfrac{3}{x^3} \vspace{1mm}\\ 0 & 1-\dfrac{2}{x^2}
\end{matrix}\right).
\end{gather*} This operator has eigenfunction $\psi$ with
eigenvalue $p$:
\begin{gather*}
 L_p\psi=L_p\circ K \psi_0=K\circ
L_0\psi_0=K\psi_0p(z)=\psi p(z).
\end{gather*} (As mentioned in
Remark~\ref{rem:eigenchange}, multiplying $\psi$ by a matrix on the
right has the ef\/fect of conjugating the eigenvalue.  It is therefore
worth noting that the eigenvalue~$p$ here is not only non-diagonal but
in fact non-diagonalizable.)

To produce an operator in $z$ sharing $\psi$ as eigenfunction, we
consider the polynomial $ p_0(z)=(z-1)^3 $ (which has this form because
$\lambda=1$ is the only point in the support of the distributions in
$\bdelta$ and because the highest derivative they take there is
$m_{\lambda}=2$).  Then
\begin{gather*}
L_0=\bisp^{-1}(p_0(z)I)=\big(\partial^3-3\partial^2+3\partial-1\big)I+\big({-}3\partial^3+6\partial^2-3\partial\big)U
\end{gather*}
 is the operator satisfying
$L_0\psi_0=\psi_0p_0(z)$ and $L_0$ factors as $L_0=Q\circ K$ with
\begin{gather*}
Q=\left(\begin{matrix}\partial+\dfrac{1-x}{x}&-3\partial+\dfrac{2x-3}{x}\vspace{1mm}\\
0&\partial+\dfrac{1-x}{x}\end{matrix}\right).
\end{gather*}

The next step is to  choose two functions from $\Cent[x]$, so that $\bar
K=\pi_K K$ and $\bar Q=Q\circ \pi_Q$ are in $\C[x,\partial]$.  It turns
out that the selection
\begin{gather*}
\pi_K=\left(\begin{matrix}x&x^2\\0&x\end{matrix}\right)= \pi_Q
\end{gather*}
works
and we get that
\begin{gather*}
 \Lambda=(zH)^2\circ(p_0(z))^{-1}\circ \bisp(\bar
Q)\circ \bisp(\bar K)\circ (zH)^{-2} = \partial_z^3 I +
\partial_z^2\left(\begin{matrix} 1 & -\dfrac{2
\left(z^2-z+6\right)}{(z-1) z} \\ 0 & 1 \end{matrix}\right) \\
\hphantom{\Lambda=}{} +
\partial_z\left(\begin{matrix}  \dfrac{4}{z-z^2} & \dfrac{2 \left(z^2+8
z+6\right)}{(z-1)^2 z^2} \vspace{1mm}\\ 0 & \dfrac{4}{z-z^2} \end{matrix}\right) +
\left(\begin{matrix}\dfrac{-2 z^2+4 z+2}{(z-1)^2 z^2} & \dfrac{4 z^3-6
z^2-8 z-20}{(z-1)^3   z^2} \vspace{1mm}\\ 0 & \dfrac{-2 z^2+4 z+2}{(z-1)^2 z^2}
\end{matrix}\right)
\end{gather*} does indeed satisfy
$\psi\Lambda=\pi_Q(x)\pi_K(x)\psi$.

\subsection[Computing $\Trans{\alpha}^H$ for non-diagonal $H$]{Computing $\boldsymbol{\Trans{\alpha}^H}$ for non-diagonal $\boldsymbol{H}$}\label{TransHexamp}

Suppose $ H=\lambda I +U $ where $U$ is still the upper-triangular
matrix from the previous example. Since
\begin{gather*}
\hbox{exp}\big(xH^{-1}\big)=e^{x/\lambda}I-\frac{x}{\lambda^2}e^{x/\lambda}U,
\end{gather*}
the corresponding operator formed by replacing $e^{ x /\lambda}$ with
$\Trans{\alpha /\lambda}$ and other occurrences of $x$ with~$\alpha\partial_z$ would be
\begin{gather*}
\Trans{\alpha}^H=\Trans{\alpha/\lambda}I-\frac{1}{\lambda^2}
\Trans{\alpha/\lambda}\alpha\partial_z U.
\end{gather*} This operator in $z$ that
combines dif\/ferentiation and translation has the property that
$(e^{xzH})\Trans{\alpha}^H$ $=e^{\alpha x}e^{xzH}$.

\subsection[An exponential example with $H=I$]{An exponential example with $\boldsymbol{H=I}$}

Finally, consider  $H=I$ and
\begin{gather*} \bdelta=\operatorname{span}\left\{
\left(\begin{matrix} \Delta_{0,0}+\Delta_{0,1}\\\Delta_{0,1}
\end{matrix}\right), \left(\begin{matrix}0\\ -\Delta_{0,0}+\Delta_{0,1}
\end{matrix}\right)\right\} .
\end{gather*}
 Then
 \begin{gather*}
K=\left(\begin{matrix}\partial-\dfrac{e^x}{1+e^x}&0\vspace{1mm}\\
\dfrac{e^{x}}{e^{2x}-1}&\partial+\dfrac{e^x}{1-e^x}\end{matrix}\right)
\qquad\text{and}\\
 \psi(x,z)=\left(\begin{matrix} \dfrac{e^{x z}
\left(e^x (z-1)+z\right)}{\left(1+e^x\right) z} & 0 \vspace{1mm}\\ \dfrac{e^{z
x+x}}{\left(-1+e^{2 x}\right) z} & \dfrac{e^{x z} \left(e^x
(z-1)-z\right)}{\left(-1+e^x\right) z} \end{matrix}\right).
\end{gather*}

Because of the support of the distributions, $p_0(z)=z^2-z$ and we may
choose $p\in\A$ to be a multiple of that by a constant matrix
\begin{gather*}
p(z)=\left(\begin{matrix}0&z^2-z\\ z^2-z&0\end{matrix}\right).
\end{gather*}
 The
constant operator $\bisp^{-1}(p)=p(\partial)$, satisf\/ies the
intertwining relationship $K\circ p(\partial)=L_p\circ K$ with
\begin{gather*}
L_p =  \left( \begin{matrix} -\dfrac{e^x \left(-1+e^x-2
e^{2 x}\right)}{\left(-1+e^{2 x}\right)^2} & -\dfrac{2 e^{2
x}}{\left(-1+e^x\right)^2 \left(1+e^x\right)} \vspace{1mm}\\ \dfrac{e^{2 x}
\left(-3+2 e^x\right)}{\left(-1+e^{2 x}\right)^2} &
\dfrac{e^x}{\left(-1+e^x\right)^2 \left(1+e^x\right)}   \end{matrix}
\right)\\
\hphantom{L_p =}{}
+ \left( \begin{matrix} -\dfrac{e^x}{-1+e^{2 x}} &
\dfrac{-1-e^x+3 e^{2 x}-e^{3 x}}{\left(-1+e^x\right)^2
\left(1+e^x\right)} \vspace{1mm}\\ \dfrac{-1+2 e^x+2 e^{2 x}-2 e^{3 x}-e^{4
x}}{\left(-1+e^{2 x}\right)^2} & \dfrac{e^x}{\left(-1+e^x\right)
\left(1+e^x\right)}   \end{matrix} \right)\partial+ \left(
\begin{matrix} 0 & 1 \\ 1 & 0 \\ \end{matrix}\right)\partial_x^2.
\end{gather*} Hence $ L_p\psi=\psi p(z). $ (Since $p(z)$ and $\psi$ do
not commute, it is necessary to write the eigenvalue on the right rather
than the left: $L\psi\not=p(z)\psi$.)

 Now, to compute an operator in $z$ having the same function $\psi$ as
 an eigenfunction, we factor $p_0(\partial)I$ as $Q\circ K$ to obtain
 \begin{gather*}
 Q=\partial I + \left( \begin{matrix} -\frac{1}{1+e^x} & 0 \\
 \frac{e^x}{1-e^{2 x}} & \frac{1}{-1+e^x}   \end{matrix} \right).
 \end{gather*}
 Letting $\pi_K$ and $\pi_Q$ be
 \begin{gather*} \pi_K=\left( \begin{matrix} e^{2 x}-1
 & 0 \\ 0 & 1-e^{2 x}   \end{matrix} \right), \qquad \pi_Q=\left(
 \begin{matrix} e^{2 x}-1 & 0 \\ 0 & e^{2 x} -1  \end{matrix} \right)
\end{gather*} we get
\begin{gather*} \bar K = \pi_KK=\left( \begin{matrix} -e^x
 \left(-1+e^x\right) & 0 \\ -e^x & e^x+e^{2 x}   \end{matrix} \right)+
 \left( \begin{matrix} \left(-1+e^x\right) \left(1+e^x\right) & 0 \\ 0 &
 1-e^{2 x}   \end{matrix} \right)\partial \in\Weyl
 \end{gather*}
  and
  \begin{gather*} \bar Q =
 Q\circ \pi_Q = \left( \begin{matrix} 1-e^x+2 e^{2 x} & 0 \\ -e^x &
 1+e^x+2 e^{2 x}   \end{matrix} \right)+ \left( \begin{matrix} -1+e^{2
 x} & 0 \\ 0 & -1+e^{2 x}   \end{matrix} \right)\partial\in\Weyl.
 \end{gather*}

Replacing $e^{\alpha x}$ by $\Trans{\alpha}^H=\Trans{\alpha}$ and
$\partial$ by $z$ we get the corresponding translation operators
\begin{gather*}
\bisp(\bar K)=\left( \begin{matrix} -z & 0 \\ 0 & z   \end{matrix}
\right)+\Trans1 \left( \begin{matrix} 1 & 0 \\ -1 & 1   \end{matrix}
\right)+\Trans2 \left( \begin{matrix} z-1 & 0 \\ 0 & 1-z   \end{matrix}
\right)
\end{gather*} and
\begin{gather*}
 \bisp(\bar Q)=\left( \begin{matrix} 1-z & 0 \\ 0 & 1-z
  \end{matrix} \right)+\Trans1 \left( \begin{matrix} -1 & 0 \\ -1 & 1
  \end{matrix} \right)+\Trans2 \left( \begin{matrix} z+2 & 0 \\ 0 & z+2
  \end{matrix} \right).
  \end{gather*}  As predicted by Theorem~\ref{mainresult},
\begin{gather*}
 \Lambda = \frac{z}{p_0(z)}\circ \bisp(\bar Q)\circ
\bisp(\bar K)\circ \frac{1}{z}   =  \left( \begin{matrix} \dfrac{z^3+5
z^2+6 z}{z (z+2) (z+3)} & 0 \vspace{1mm}\\ 0 & \dfrac{-z^3-5 z^2-6 z}{z (z+2) (z+3)}
  \end{matrix} \right)\\
 \hphantom{\Lambda =}{}
  +\Trans1 \left( \begin{matrix} 0 & 0 \vspace{1mm}\\ \dfrac{2
z^2+6 z+4}{z (z+1) (z+2)} & 0   \end{matrix} \right)
+\Trans2 \left(
\begin{matrix} \dfrac{-2 z^3-10 z^2-12 z}{z (z+2) (z+3)} & 0 \vspace{1mm}\\ \dfrac{-2
z-4}{z (z+1) (z+2)} & \dfrac{2 z^3+10 z^2+12 z}{z (z+2) (z+3)}
\end{matrix} \right)\\
 \hphantom{\Lambda =}{}
+\Trans3 \left( \begin{matrix} \dfrac{4 z+12}{z (z+2)
(z+3)} & 0 \vspace{1mm}\\ \dfrac{-2 z^2-4 z-2}{z (z+1) (z+2)} & \dfrac{4 z+12}{z (z+2)
(z+3)} \\ \end{matrix} \right)\\
\hphantom{\Lambda =}{}
+\Trans4 \left( \begin{matrix}
\dfrac{z^3+5 z^2+2 z-8}{z (z+2) (z+3)} & 0 \vspace{1mm}\\ 0 & \dfrac{-z^3-5 z^2-2
z+8}{z (z+2) (z+3)}   \end{matrix} \right)
\end{gather*} satisf\/ies
$\psi\Lambda=\piQK\psi$ ($\piQK=\pi_Q\pi_K$).

\begin{Remark}\label{rem:ad-examp} This is a good example for demonstrating
ad-nilpotency and how it has changed in this non-commutative context.
Using $L$, $\Lambda$, $p(z)$ and $\pi_K$, $\pi_Q$ as above, it is indeed
true that
\begin{gather*} \ad_{\piQK}^nL\psi=\psi\ad_{\Lambda}^np
\end{gather*} (with the
operator in $z$ acting from the right as usual).  However, contrary to
our expectation from the commutative case in which the order of the
operator on the left decreases to zero when $n$ gets large,
$\ad_{\piQK}^nL$ is an operator of order~$2$ for every~$n$. This happens
because the  action of iterating~$\ad_{\piQK}$ on the leading
coef\/f\/icient of~$L$  itself is not nilpotent. On the other hand, if
instead of~$\pi_K$ and~$\pi_Q$ we had chosen
\begin{gather*}
\pi_K^*=\pi_Q^*=\big(e^{2x}-1\big)I,\qquad \piQK^*=\pi_K^*\pi_Q^*=\big(e^{2x}-1\big)^2I
\end{gather*} (intentionally selected so as to commute with the coef\/f\/icients of any
operator) then  Theorem~\ref{mainresult} would have produced a dif\/ferent
operator $\Lambda^*$ satisfying $\piQK^*\psi=\psi\Lambda^*$.  In this
case, because the order is lowered at each iteration,
$\ad_{\piQK^*}^3L=0$.  This is not particularly surprising or special as
the same could be said if $\piQK^*$ were replaced by any function of the
form $f(x)I$. However, because of the correspondence in
Theorem~\ref{ad}, we can conclude from this that
$\ad_{\Lambda^*}^3p(z)=0$ also. That is more interesting because in
general repeatedly taking the commutator of the operator~$\Lambda^*$
with some function will not produce the zero operator, even if all
coef\/f\/icients are assumed to commute. Nevertheless, with these specif\/ic
choices everything cancels out leaving exactly zero when
$\ad_{\Lambda^*}$ is applied   three  times to the  function $p(z)$ that
is the eigenvalue above.
\end{Remark}

\section{Conclusions and remarks}

Given a choice of an invertible $N\times N$ matrix $H$ and
$MN$-dimensional space $\bdelta$ of vector-valued f\/initely supported
distributions, this paper sought to produce a bispectral triple
$(L,\Lambda,\psi)$ where $L$ is a dif\/ferential operator acting from the
left, $\Lambda$ is a dif\/ferential-translation operator acting from the
right and $\psi$ is a common eigenfunction that is asymptotically of the
form $e^{xzH}$ satisfying the ``conditions'' generated by the
distributions.  In the scalar case, this was achieved in~\cite{BispSol}
for \textit{any} choice of $\bdelta$.  In the matrix generalization
above, however, the construction only works given
Assumptions~\ref{assumpA} and \ref{assumpB}. The bispectral triples
produced given those two assumptions include many new examples both in
the form of the eigenfunction (asymptotically equal to~$e^{xzH}$) and
the fact that the matrix-coef\/f\/icient operator~$\Lambda$ may involve
translation in $z$ as well as dif\/ferentiation in $z$.  More importantly,
this investigation yielded some observations that may be useful in
future studies of bispectrality in a non-commutative context.

This paper sought  to develop a general construction of bispectral
triples with matrix-valued eigenvalues  and also to understand what
obstructions there might be to generalizing the construction from
\cite{BispSol} to the matrix case.  It is interesting to note that these
seemingly separate goals both turn out to depend on the
non-commutativity of the ring~$\A$ of functions that stabilize the point
in the Grassmannian. Since~$\A$ is also the ring of eigenvalues for the
operators in~$x$, it is not a surprise that by letting its elements be
matrix-valued gives us a non-commutative ring $\{L_p\colon p\in\A\}$ of
operators sharing the eigenfunction $\psi$.  It was not clear at f\/irst
that the ring $\A$ would also be the source of the obstruction to
producing bispectral triples.  However, the construction of operators in
$z$ sharing $\psi$ as eigenfunction uses the assumption that~$H$ is an
element of $\A$.  (Specif\/ically, this is used in the proof of
Lemma~\ref{KcommuteswithH}.)  It is interesting to note that this
\textit{also} depends on the fact that we considered a matrix-valued
stabilizer ring.

Section~\ref{sec:ad} explored the extent to which the property of
ad-nilpotency, which has been a~feature of papers on the bispectral
problem since it was f\/irst noted in~\cite{DG}, continues to apply in the
case of matrix coef\/f\/icient operators.  It is still the case that the
operator formed by iterating the adjoint of one of the operators in a
bispectral triple on the eigenvalue of the other operator has the same
action on the eigenfunction as the operator formed by iterating the
adjoint action of its eigenvalue on the other operator.  However, unlike
the scalar case, only if additional assumptions about the coef\/f\/icients
of the operators are met will either of those be zero for a~large enough
iterations.

It was previously observed \cite{BGK,BL,WilsonNotes} (see also \cite{matrix3, GPT1,GPT2,GPT3,GPT4}) that requiring the operators in
$x$ and $z$ to act on the eigenfunctions from opposite side resulted in
a form of bispectrality whose structure and applications more closely
resembled that in the scalar case.  Because generalizing the scalar
results of \cite{BispSol} necessitated also requiring that the
eigenvalues and f\/initely-supported distributions act from the same side
as the operators acting in the same variables, this previous observation
can now be extended to those other objects as well.

This is the f\/irst time that the bispectral anti-isomorphism was used
to construct bispectral operators in a context involving operators
with matrix coef\/f\/icients, see section~\ref{sec:antiiso} (cf.~\cite{GHY}).  Some
modif\/ications were necessary since the rather general construction in~\cite{BHY} assumed that there were no zero-divisors so that formal
inverses could be introduced.  In addition, the use of the method here
depended on the convention of considering operators in~$x$ and~$z$ to
be acting from opposite sides and required that the coef\/f\/icients of
the operators on which it acted were taken from the centralizer of~$H$.  The most interesting dif\/ference may have been that because the
operators in~$z$ are acting from the right rather than the left, the
map actually \textit{preserves} the order of a product.

Unlike the scalar case, not every choice of distributions $\bdelta$
corresponds to a bispectral triple.  For instance, if
Assumption~\ref{assumpA} fails to be met then there simply is no wave
function $\psi$ satis\-fying the conditions.  More interestingly, if
Assumption~\ref{assumpA} is met but Assumption~\ref{assumpB} is not then
there is a wave function that is an eigenfunction for a ring of
dif\/ferential operators in~$x$, but Theorem~\ref{mainresult} does not
produce a corresponding operator in~$z$.  In fact, as Section~\ref{nogo}
demonstrates, in at least some cases there actually is no bi\-spectral
triple of the form considered above which includes that wave function.  This is very dif\/ferent than the scalar case. Unfortunately, this
paper does not entirely answer the question of which choices of matrix~$H$ and distributions~$\bdelta$ produce a wave function $\psi$ that is
part of a bi\-spectral triple of the type considered here.  In particular,
this paper does not show or claim that there \textit{cannot be} a~dif\/ferential-dif\/ference operator~$\Lambda$ in~$z$ having~$\psi$ as an
eigenfunction with $x$-dependent eigenvalue when~$H\not\in\A$.

Arguably, some of the non-commutativity involved in the construction
above was ``artif\/icially inserted'' in the form of the choice of the
matrix $H$.  If one chooses to consider only the case $H=I$, then
Assumption~\ref{assumpB} is automatically met and Theorems~\ref{claimLp}
and~\ref{mainresult} produce a~bi\-spectral triple for any~$\bdelta$
satisfying Assumption~\ref{assumpA}.  Since any operator $L_p$ produced
by the construction above for some choice of~$H$ \textit{can} be
produced using a dif\/ferent choice of distributions but with $H=I$, one
may conclude that each of these operators is part of a bispectral
triple.  (In other words, the obstruction to bi\-spectrality that is
visible when one seeks a bi\-spectral triple for a given wave function~$\psi$ disappears if one instead focuses on the operator~$L_p$ and seeks
a corresponding bispectral triple.)
However, it seems plausible that some examples of bispectrality to be considered in the future will involve vacuum eigenfunctions that are necessarily non-commutative (unlike these examples in which the non-commutativity of the vacuum eigenfunctions can always be eliminated through a change of variables), and the observations and results above will prove useful in those contexts.

\subsection*{Acknowledgements}

The author thanks the College of Charleston for
the sabbatical during which this work was completed, Maarten Bergvelt
and Michael Gekhtman for mathematical assistance as well as serving as
gracious hosts, Chunxia Li for carefully reading and commenting on
early drafts, and the referees for their advices.

\pdfbookmark[1]{References}{ref}
\LastPageEnding

\end{document}